\pdfoutput=1
\documentclass{article}

\usepackage{arxiv}

\usepackage{breakurl}             
\usepackage{underscore}           

\usepackage{stmaryrd,amsmath,amssymb,graphicx,xcolor}

\newcommand{\tbox}[1]{\hbox{\it #1}}

\newcommand{\tempop}[1]{\mathop{\tbox{#1}\,}}

\newcommand{\temprel}[1]{\mathrel{\tbox{#1}}}
\newcommand{\tempord}[1]{\mathord{\tbox{#1}}}

\newcommand{\ITLImplies}{\mathrel{\Rightarrow}}

\newcommand{\ITLEquiv}{\mathrel{\Leftrightarrow}}

\newcommand{\ITLSat}{\vDash}

\newcommand{\ITLdefs}{\mathrel{\triangleq}}
\newcommand{\ITLDefs}{\mathrel{\triangleq}}

\newcommand{\ITLTheorem}{\vdash\quad}
\newcommand{\ITLtheorem}{\vdash}
\newcommand{\ITLTrue}{\tempord{\sf true}}

\newcommand{\ITLFalse}{\tempord{\sf false}}

\newcommand{\ITLNot}{\neg}
\newcommand{\ITLAnd}{\mathrel{\wedge}}
\newcommand{\ITLOr}{\mathrel{\vee}}

\newcommand{\itlNextsym}{\raise 0.20em\hbox{$\scriptstyle \bigcirc\mskip -2.5mu$}}
\newcommand{\ITLNext}{\mathop{\itlNextsym}}
\newcommand{\ITLAlways}{\mathop{\Box}}
\newcommand{\ITLSometime}{\mathop{\Diamond}}
\newcommand{\ITLspot}{\mathrel{\setlength{\unitlength}{1ex}\begin{picture}(1,1)
      \put(0.5,0.6){\circle*{0.5}}\end{picture}}}

\newcommand{\itlADiamond}[1]{\mathop{\settowidth{\unitlength}{$\ITLSometime$}\begin{picture}(1,1)
  \put(0,0){\makebox(1,1)[b]{$\ITLSometime$}}\put(0,.4){\makebox(1,1)[b]{\resizebox{!}{.45\unitlength}{\makebox(1,1)[c]{#1}}}}\end{picture}}}

\newcommand{\itlABox}[1]{\mathrel{\settowidth{\unitlength}{$\ITLAlways$}\begin{picture}(1,1)
  \put(0,0){\makebox(1,1)[b]{$\ITLAlways$}}\put(0,0.01){\makebox(1,1)[b]{\resizebox{!}{.90\unitlength}{\makebox(1,1)[c]{\scriptsize
          #1}}}}\end{picture}}}

\newcommand{\itlANext}[1]{\mathrel{\settowidth{\unitlength}{$\ITLNext$}\begin{picture}(1,1)
      \put(0,0){\makebox(1,1)[b]{$\ITLNext$}}\put(0.1,0){\makebox(1,1)[c]{\resizebox{!}{.45\unitlength}{\makebox(1,1)[c]{#1}}}}\end{picture}}}

\newcommand{\ITLPrev}{\itlANext{\textbf{--}}}
\newcommand{\ITLWeakPrev}{\itlANext{{$\sim$}}}
\newcommand{\ITLWeakNext}{\itlANext{w}}

\newcommand{\ITLDa}{\itlADiamond{\textsf{a}}}
\newcommand{\ITLBa}{\itlABox{\textsf{a}}}

\newcommand{\ITLDi}{\itlADiamond{\textsf{i}}}
\newcommand{\ITLBi}{\itlABox{\textsf{i}}}

\newcommand{\ITLExists}[1]{{\exists #1}\ITLspot}

\newcommand{\ITLKeep}{\tempop{\sf keep}}

\newcommand{\ITLHalt}{\tempop{\sf halt}}

\newcommand{\ITLFin}{\tempop{\sf fin}}
\newcommand{\ITLInit}{\tempop{\sf init}}
\newcommand{\ITLInitonly}{\tempop{\sf initonly}}

\newcommand{\ITLGets}{\temprel{\sf gets}}
\newcommand{\ITLLen}{\tempop{\sf len}}

\newcommand{\ITLChop}{\mathbin{;}}

\newcommand{\ITLChopstar}{{}^{*}}

\newcommand{\ITLSkip}{\tempord{\sf skip}}
\newcommand{\ITLEmpty}{\tempord{\sf empty}}
\newcommand{\ITLMore}{\tempord{\sf more}}

\newcommand{\ITLWhile}{\tempop{\sf while}}

\newcommand{\ITLDo}{\temprel{\sf do}}
\newcommand{\ITLIf}{\tempop{\sf if}}

\newcommand{\ITLThen}{\temprel{\sf then}}
\newcommand{\ITLElse}{\temprel{\sf else}}

\newcommand{\Tassign}{\leftarrow}

\newcommand{\STrue}{\mbox{tt}}

\newcommand{\SFalse}{\mbox{ff}}

\newcommand{\SNot}{\mbox{not }}

\newcommand{\SAnd}{\mbox{ and }}

\newcommand{\SIff}{\mbox{ iff }}

\newcommand{\intlen}[1]{{|#1|}}
\newcommand{\intmap}[2]{\mathsf{map}\,\, {#1}\,\, {#2}} 
\newcommand{\Sem}[2]{\mathsf{M}\llbracket{#1}\rrbracket(#2)}

\newcommand{\ITLRev}{{}^{r}}
\newcommand{\intrev}[1]{\mathsf{rev}(#1)}
\newcommand{\eSem}[2]{\mathsf{E}\llbracket{#1}\rrbracket(#2)}

\newcommand{\zet}{{\mathchoice {\hbox{$\sf\textstyle Z\kern-0.4em Z$}}
{\hbox{$\sf\textstyle Z\kern-0.4em Z$}}
{\hbox{$\sf\scriptstyle Z\kern-0.3em Z$}}
{\hbox{$\sf\scriptscriptstyle Z\kern-0.2em Z$}}}}

\newtheorem{example}{Example} 
\newtheorem{theorem}{Theorem}
\newtheorem{proof}{Proof} 
\newtheorem{definition}{Definition}
\newtheorem{lemma}{Lemma}

\begin{document}

\title{Reversibility of
Executable \\ Interval Temporal Logic  Specifications
\thanks{Supported by DMU.}}


\author{\textbf{Antonio Cau} \\
Cyber Technology Institute,\\
School of Computer Science and Informatics, \\
Faculty of Computing, Engineering and Media\\
De Montfort University, Leicester, UK \\
\texttt{antonio.cau@dmu.ac.uk}
\and \textbf{Stefan Kuhn} \\
Cyber Technology Institute,\\
School of Computer Science and Informatics, \\
Faculty of Computing, Engineering and Media\\
De Montfort University, Leicester, UK \\
\texttt{stefan.kuhn@dmu.ac.uk}
\and \textbf{James Hoey}\\
School of Informatics,\\
University of Leicester, Leicester, UK \\
\texttt{jbh11@leicester.ac.uk}
}



%

\maketitle              
%
\begin{abstract}

In this paper the reversibility of executable Interval Temporal Logic (ITL) specifications is investigated. ITL allows for the reasoning about systems in terms of behaviours which are represented as non-empty sequences of states. It allows for the specification of systems at different levels of abstraction. At a high level this specification is in  terms of properties, for instance safety and liveness properties. At concrete level one can specify a system in terms of programming constructs. One can execute these concrete specification, i.e., test and simulate the behaviour of the system. In this paper we will formalise this notion of executability of ITL specifications.
ITL also has a reflection operator which allows for the reasoning about reversed behaviours. We will investigate the reversibility of executable ITL specifications, i.e., how one can use this reflection operator to reverse the concrete behaviour of a particular system.


\end{abstract}
\section{Introduction}

Formal methods have been used in computer science to verify desirable and undesirable properties of programs. One type of formalism introduced is temporal logic. A temporal logic allows to reason about properties over time, for example ``this resource will eventually be freed''. In this paper, we are dealing with a particular temporal logic, Interval Temporal Logic (ITL).

Another strand of research is reversibility in computing. This is relevant for reversing the effects of operations, for example if, after having performed a number of operations, proceeding in the desired direction is not possible. This could be because a resource is not available or because a result is outside the allowed range of values. In such cases, a potential strategy is to roll back to a safe state and continue operation from there.

Using the ITL notation (details of which will explained in Section~\ref{sec:itl}), a program consisting of two parts could be written as $\mathsf{Good} \ITLChop \mathsf{Bad}$. The semantics (behaviour) of both parts are sequences of states. $\mathsf{Good}$ and $\mathsf{Bad}$ are arbitrary names indicating sections of the program which worked as expected respectively did not. We now want to reverse the effect of $\mathsf{Bad}$. This would require that we go back to the last state of $\mathsf{Good}$. This can be done by using an operator~$\mathsf{undo}$ so that we have  $\mathsf{Good}\ITLChop\mathsf{Bad}\ITLChop\mathsf{undo}(\mathsf{Bad})$. This operator must ensure that the last state of  $\mathsf{Good}$ is the same as the last state of  $\mathsf{Good}\ITLChop\mathsf{Bad}\ITLChop\mathsf{undo}(\mathsf{Bad})$. 

We propose a solution for this problem, where we use the reflection operator~$\ITLRev$. We show that this reverses the effects of a formula, i.e., reversing the sequence of states of that formula. We also show that it can be applied to any  formula that can be specified in ITL. This operator can be used for propositional as well as for first-order logic ITL. Therefore, we have a universal undo operation for these formulae. We distinguish \textit{reflection}, which indicates the possibility to reverse the sequences of states, from \textit{reversibility}. Reversibility indicates that an executable formula, a program, can be reversed. 

The outline of the paper is as follows: in Section~\ref{sec:background} we discuss temporal logic in general and compare ITL with other temporal logics. We also discuss reversibility in general. In Section~\ref{sec:itl} we discuss ITL, i.e., intervals, syntax of basic and derived constructs, and the reflection operator and the semantics of these constructs. In Section~\ref{sec:execrefrev}, the notion of executability is formalised and show how this notion can be used together with reflection to reverse the effects of bad computations. 
In Section~\ref{sec:results} we summarize our results followed by a conclusion and future work in Section~\ref{sec:conclusion}.




\section{Background}
\label{sec:background}
Propositional logic deals with propositions and their connections by propositional connectives. This allows reasoning about truth of propositions. An extension of propositional logic is First-order logic, which uses quantifiers to reason about objects in a domain. These are powerful tools, allowing for example reasoning in knowledge-based systems like expert systems.

Such logic systems found applications in computer science and software engineering, mainly for specifying transformational systems. Such a system takes some input, transforms it by doing some calculation, and produces an output, after which it terminates. In contrast, reactive systems are non-terminating, consume input, and produce output continuously. To describe and analyse such a system, the temporal aspect is important, and consequently, temporal extensions of logic have been developed.

After some cursory mentions earlier, a first type of temporal logic was presented by N. Prior \cite{mates1956}. Based on this, other types of temporal logic were devised, including  (LTL) \cite{4567924,MannaPnueli1992}. 
The main operator of LTL is the until ($\mathcal{U}$) while in ITL it is the chop ($\ITLChop$). $f \mathbin{\mathcal{U}} g$  guarantees that $g$ will eventually hold at some future state and that $f$ will continue to hold until then. 
In ITL, satisfaction of formulas is defined over intervals (non-empty sequences of states) rather than time points which is used in LTL. 
$f \ITLChop g$ denotes that the interval can be split into a prefix and a suffix interval in such a way that $f$ holds for the prefix interval and $g$ holds for the suffix interval. So the chop operator corresponds to the sequential composition operator. 

Temporal variables are used e.g. in TLA~\cite{TLA}, where primed variables which denote the value of a variable in the next state. 
In ITL we have temporal variables for values of variables in the next state, the penultimate state and the final state. 


Whilst time is generally considered irreversible and reversibility in the physical world is only possible at the cost of increasing entropy, there has been increasing interest and research done in reversible computation, for example demonstrated by the COST action IC1405 \cite{costbook}. We disregard the implications of reversible computing for hardware design and its potential for energy conservation here, but focus on logical reversibility. This is all types of formalisms which allow reversing steps done in order to get back to a previous state of the computation \cite{5391327}.

Any temporal logic models computations over time, and time is generally irreversible, but computations can be reversible, as we have seen. Because of this, it seems a logic extension of temporal logic to introduce a ``time reversal'', which undoes computations and therefore seemingly reverses time, whereas actual time is progressing. \cite{Moszkowski2014} introduced this for propositional ITL. In this paper, we extend this to first-order logic ITL.

Reversibility and reflection of logic is related to reversibility of programming languages. Many works have researched the process of reversing executions of traditional programming languages, most of which are typically irreversible as information is lost throughout. One approach to reversing such executions is to save this lost information as a program executes forward and later use it to reconstruct previous states (reflection). This includes the Reverse C Compiler \cite{KP2014} and the works of Hoey and Ulidowski \cite{JH2018,JH2019}. Any irreversible step of an execution is made reversible via this saved information. Execution time and memory usage are crucial aspects of these methods, with a forward execution typically being slower and memory requirements being higher as information is recorded. Such approaches including that described here minimise these overheads sufficiently. This differs from checkpointing approaches, where a snapshot of the state is taken at regular intervals and used to restore to previous positions \cite{KP2014}. Depending on the snapshot frequency, large amounts of information must be recorded and forward re-execution is sometimes required. A second approach is to use reversible languages such as Janus \cite{CL1986}, where any valid program written in such a language can be executed both forward and in reverse. This is comparable to ITL programs whose reflections are executable. Janus relies on the use of increment/decrement operators to ensure no old values of variables are lost, as well as post-conditions that allow correct expression evaluation during a reverse execution. However the challenge of converting programs of a traditional language into that of a reversible language may limit its widespread use. 

Reversibility has many interesting applications \cite{costbook}. This includes debugging \cite{JE2012,EG2014,debugCaseStudy}, where code defects can be located and fixed by executing a misbehaving program in reverse, and discrete event simulation \cite{CC1999,MS2018}, where optimistic execution can be rolled back when required. 





\section{Interval Temporal Logic}
\label{sec:itl}
Interval Temporal Logic (ITL) is a flexible notation for both propositional and first-order reasoning about periods of time found in descriptions of hardware and software systems~\cite{zhou05,Cauhomepage}. Unlike most temporal logics, ITL can handle both sequential and parallel composition and offers powerful and extensible specification and proof techniques for reasoning about properties involving safety, liveness and projected time. Timing constraints are expressible and furthermore most imperative programming constructs can be viewed as formulas in ITL. AnaTempura (available from \cite{Cauhomepage}) provides an executable framework for developing and experimenting with suitable ITL specifications.

\subsection{Interval}
\label{sec:model}
In this section we revisit the underlying semantic model of Interval Temporal Logic (albeit restricted to the finite case).

The key notion of ITL is an \textit{interval}.  An interval $\sigma$ is considered to be a non-empty, finite sequence of states $\sigma_0, \sigma_1\ldots, \sigma_n$. A state is the union of an integer state $\mathsf{State}^e$ which is a mapping from the set of integer variables $\mathsf{Var}^e$ to the set of integer values $\mathsf{Val}$, and a Boolean state $\mathsf{State}^b$ which is a mapping from the set of propositional variable $\mathsf{Var}^b$ to the set of Boolean values $\mathsf{Bool}$. Note: the embedding of ITL in Isabelle/HOL is such that one can use any definable type in Isabelle/HOL as type for an ITL variable. We have restricted the types to just integers and Boolean in this paper. Let $\Sigma^+$ denote the set of all finite intervals with at least 1 state. The length of an interval $\sigma$ is denoted by $\intlen{\sigma}$ and is the number of states minus 1, i.e., an interval with one state has length zero. Let $\sigma=\sigma_0 \sigma_1 \sigma_2 \ldots \sigma_{\intlen{\sigma}}$ be an interval then $\sigma_0\ldots \sigma_k (\mbox{where } 0\leq k \leq \intlen{\sigma})$ denotes a prefix interval of $\sigma$, $\sigma_k\ldots \sigma_{\intlen{\sigma}} (\mbox{where } 0\leq k \leq \intlen{\sigma}) $ denotes a suffix interval of $\sigma$,  $\sigma_k\ldots \sigma_l \;\; (\mbox{where } 0\leq k \leq l \leq \intlen{\sigma}) $ denotes a sub interval of $\sigma$.
 
%
    

\subsection{Syntax}
\label{sec:syntax}
We first discuss the basic constructs and then introduce derived constructs. 

We introduce the basic constructs: propositional and integer temporal variables, $\ITLTrue$, $\ITLSkip$, $\ITLChop$ (chop), $\ITLChopstar$ (chopstar) operators. 
 
Syntax of Integer Expressions in BNF:
\[\begin{array}{rl}
   ie ::= & z \; | \; 
	    ig(ie_1,\ldots,ie_n) \; | \;
   A \; | \; \ITLFin A \; | \; \ITLNext A  
\end{array}
\]
where $z$ is an integer constant, $ig$ an integer operator, and $A$ , $\ITLNext A$ and $\ITLFin A$ are temporal variables.

Syntax of Boolean Expressions in BNF:
\[\begin{array}{rl}
   be ::= & b \; | \; 
	    bg(be_1,\ldots,be_n) \; | \;
   Q \; | \; \ITLFin Q \; | \; \ITLNext Q 
\end{array}
\]
where $b$ is Boolean constant, $bg$ a Boolean operator, and $Q$, $\ITLNext Q$ and $\ITLFin Q$ are temporal variables.

Syntax of First Order Formulae in BNF:
\[\begin{array}{rl}
   f ::= & \ITLTrue \; | \; h(e_1,\ldots,e_n) \; | \; 
              \ITLNot f \; | \; f_1 \ITLAnd f_2  \; | \; \ITLExists{V} f \; | \; 
          \ITLSkip \; | \; f_1 \ITLChop f_2 \; | \; f\ITLChopstar
\end{array}
\]
where $h$ is a Boolean predicate over integer or Boolean expressions, and
$V$ is a Boolean or integer variable.

The formula $\ITLSkip$ denotes any interval of exactly two states.
The formula $f_1\ITLChop f_2$, where $f_1$ and $f_2$ are ITL formulae denotes an interval which is the fusion of two intervals, $f_1$ holds over the first interval and $f_2$ holds for the second interval. Fusion will concatenate two intervals in such a way that the last state of the first interval and the first state of the second interval are ``fused'' together. Fusion is only possible when these states are the same. If these states are not the same the resulting interval does not exist, i.e., is $\ITLFalse$.
The formula $f\ITLChopstar$ where $f$ is an ITL formula denotes the fusion of a finite number of intervals, where for each interval $f$ holds. Zero times fusion will result in an interval with exactly one state irrespective of $f$, i.e., $\ITLFalse\ITLChopstar$ is equal to $\ITLEmpty$. 
Temporal variables $\ITLNext V$ and $\ITLFin V$ denote the value of variables at a particular point in an interval and are used to specify assignment constructs. The temporal variable $\ITLNext A$ denotes the value of $A$ in the next state. 
The expression $\ITLFin A$ denotes 
the value of $A$ in the last state.
The formula $\ITLExists{V} f$ denotes the introduction of a local variable $V$.


    
  
\subsubsection{Derived Constructs}  
The traditional Linear Temporal Logic (LTL) operators $\ITLNext, \ITLSometime$ and $\ITLAlways$ are defined as follows:
The formula  $\ITLNext f \ITLdefs \ITLSkip\ITLChop f$ denotes that $f$ holds from the next state. Note that $\ITLNext f$ is different from temporal variable $\ITLNext V$, although the same $\ITLNext$ symbol is used, $\ITLNext f$ is using the $\ITLNext$ symbol on formula $f$ whereas $\ITLNext V$ is using the $\ITLNext$ on a variable $V$. $\ITLNext f$ itself is a formula whereas $\ITLNext V$ denotes a value. The formula $\ITLSometime f \ITLdefs\ITLTrue \ITLChop f$ (sometimes) denotes that there exists a suffix interval for which $f$ holds. The formula  $\ITLAlways f \ITLdefs \ITLNot \ITLSometime \ITLNot f$ (always) denotes that for each suffix interval $f$ holds. The formula $\ITLMore \ITLdefs \ITLNext \ITLTrue$ denotes an interval with at least two states. The formula  $\ITLEmpty \ITLdefs \ITLNot \ITLMore$ denotes an interval with only one state. Note that no interval will satisfy the formula $\ITLFalse$.  The formula  $\ITLWeakNext f \ITLdefs \ITLEmpty \ITLOr \ITLNext f$ (weak next) denotes either an interval of only one state or $f$ holds from the next state. The formula $\ITLDi f \ITLdefs f \ITLChop \ITLTrue$ (diamond-i) denotes that there exists a prefix interval for which $f$ holds. The formula $\ITLBi f \ITLdefs \ITLNot \ITLDi \ITLNot f$ (box-i) denotes that for each prefix interval $f$ holds. The formula  $ \ITLDa f \ITLdefs \ITLTrue\ITLChop f\ITLChop \ITLTrue$ (diamond-a) denotes that there exists a sub interval for which $f$ holds. The formula  $ \ITLBa f \ITLdefs \ITLNot \ITLDa \ITLNot f$ (box-a) denotes that for each sub interval $f$ holds.

\subsection{Semantics}
We now define the semantics of ITL which is a mapping from the syntactic constructs of Section~\ref{sec:syntax} and the semantic model (intervals) defined in Section~\ref{sec:model} to values (Boolean or integers). 
Let $\eSem{\ldots }{\ldots}$ be the ``meaning'' (semantic) function from $\mathsf{Expressions} \times \Sigma^+  $ to $\mathsf{Val}$ and let $\sigma=\sigma_0\sigma_1\ldots$ be an interval then the semantics of integer expressions is as follows
\[\begin{array}{lll}
\eSem{z}{\sigma} & = & z \\
\eSem{A}{\sigma} & = & \sigma_0(A) \\
\eSem{ig(ie_1,\ldots,ie_n)}{\sigma} & = & ig(\eSem{ie_1}{\sigma},
   \ldots, \eSem{ie_n}{\sigma})  \\
\eSem{\ITLNext A}{\sigma} & = & 
     \left \{\begin{array}{ll}
       \sigma_1(A) & \mbox{if } \intlen{\sigma}>0 \\
       \mbox{choose-any-from}(\mathbb{Z}) & \mbox{ otherwise}
     \end{array}
     \right .
      \\
\eSem{\ITLFin A}{\sigma} & = & 
       \sigma_{\intlen{\sigma}}(A) 
   \\
 \end{array}
 \]
The semantics of Boolean expressions:
\[\begin{array}{lll} 
\eSem{b}{\sigma} & = & b \\
   \eSem{Q}{\sigma} & = & \sigma_0(Q)\\
\eSem{bg(be_1,\ldots,be_n)}{\sigma} & = & bg(\eSem{be_1}{\sigma},
   \ldots, \eSem{be_n}{\sigma})  \\
   \eSem{\ITLNext Q}{\sigma} & = & 
     \left \{ \begin{array}{ll}
       \sigma_1(Q) & \mbox{if } \intlen{\sigma}>0 \\
       \mbox{choose-any-from}(\text{Bool}) & \mbox{ otherwise}
     \end{array}
     \right .
      \\
\eSem{\ITLFin Q}{\sigma} & = & 
       \sigma_{\intlen{\sigma}}(Q) 
   \\
\end{array}
\]
Let $\Sem{\ldots }{\ldots}$ be the ``meaning'' function from $\mathsf{Formulae} \times \Sigma^+ $ to $\mathsf{Bool}$ (set of Boolean values, $\{\STrue,\SFalse\}$)  and let $\sigma=\sigma_0\sigma_1\ldots\sigma_{\intlen{\sigma}}$ be an interval. Let $\sigma \mathrel{\sim_V} \sigma'$ denote that the intervals $\sigma$ and $\sigma'$ are identical with the possible exception of the mapping for the variable $V$. The semantics of formulae is:
\[\begin{array}{lll}
\Sem{\ITLTrue}{\sigma} & = & \STrue \\
\Sem{h(e_1,\ldots,e_n)}{\sigma}=\STrue &
       \SIff &
  h(\eSem{e_1}{\sigma}, \ldots, \eSem{e_n}{\sigma}) \\
\Sem{\ITLNot f}{\sigma} = \STrue & \SIff  &
  \SNot (\Sem{f}{\sigma} =\STrue) \\
\Sem{f_1 \ITLAnd f_2}{\sigma} =\STrue & \SIff & 
       (\Sem{f_1}{\sigma}=\STrue) \SAnd 
       (\Sem{f_2}{\sigma}=\STrue) \\
\Sem{\ITLSkip}{\sigma} = \STrue & \SIff & \intlen{\sigma}=1\\
\Sem{\ITLExists{V} f}{\sigma}=\STrue  &\SIff & 
     (\mbox{exists } \sigma' \mbox{ s.t. } 
       \sigma \mathrel{\sim_V} \sigma' \SAnd  
       \Sem{f}{\sigma'}=\STrue)\\
       \Sem{f_1 \ITLChop f_2}{\sigma} =
\STrue & \SIff & (\mbox{exists } k, \mbox{ s.t. } \Sem{f_1}{\sigma_0\ldots\sigma_k} =
  \STrue \SAnd \\
  & & \,\,\Sem{f_2}{\sigma_k\ldots\sigma_{\intlen{\sigma}}} =
  \STrue) \\
  \Sem{f\ITLChopstar}{\sigma} =\STrue & \SIff &
  (\mbox{exist }l_0, \ldots, l_n \mbox{ s.t. } 
        l_0=0 \SAnd l_n = \intlen{\sigma} \SAnd \\
       & & \,\,\mbox{for all } 0\leq i <n, l_i < l_{i+1}
       \SAnd \\
       & & \,\,\,\,\Sem{f}{\sigma_{l_i}\ldots\sigma_{l_{i+1}}} = \STrue)\\
\end{array}
\]
A first order ITL formula $f$ is satisfiable denoted by $\ITLSat f$ if and only if there exists an interval $\sigma$ such that $\Sem{f}{\sigma} = \STrue$. A  first order ITL formula $f$ is valid denoted by $ \ITLtheorem f$ if and only if for all intervals $\sigma$, $\Sem{f}{\sigma} = \STrue$.

\subsection{Reflection}
\label{sec:reflection}

We now discuss the notion of temporal reflection for ITL formulae as defined in \cite{Moszkowski2014}. We first discuss the semantic notion of the reverse of a sequence of states and then discuss the reflection operator and its corresponding semantics. 

Let $f$ be a formula, $e$ an expression, and $\sigma$ be an interval $\sigma_0 \ldots \sigma_{\intlen{\sigma}}$ then 
$\intrev{\sigma}$ denotes interval reversal and is defined as
\[
\mathsf{rev}(\sigma) \ITLDefs \sigma_{\intlen{\sigma}} \ldots \sigma_0.\]
$f\ITLRev$ denotes temporal reflection of formula $f$ and is defined as 
\[\Sem{f\ITLRev}{\sigma} \ITLDefs \Sem{f}{\intrev{\sigma}}.
\]
$e\ITLRev$ denotes temporal reflection of expression $e$ and is defined as 
  \[\eSem{e\ITLRev}{\sigma} \ITLDefs \eSem{e}{\intrev{\sigma}}.\]
We show that for any basic operator in ITL its reflection $f\ITLRev$ is in ITL. This implies that  ITL is closed under reflection.

From \cite{Moszkowski2014} we have the following reflection laws:
\[ \begin{array}{ll|ll}
  R_0 & \quad \ITLTheorem \ITLTrue\ITLRev \ITLEquiv \ITLTrue \quad & R_1 & \quad \ITLTheorem (\ITLNot f)\ITLRev \ITLEquiv \ITLNot (f\ITLRev)\\
  R_2 & \quad \ITLTheorem (f_1 \ITLAnd f_2)\ITLRev \ITLEquiv (f_1\ITLRev) \ITLAnd (f_2\ITLRev) \quad & 
  R_3 & \quad \ITLTheorem \ITLSkip\ITLRev \ITLEquiv \ITLSkip\\
  R_4 & \quad \ITLTheorem (f_1\ITLChop f_2)\ITLRev \ITLEquiv (f_2\ITLRev\ITLChop f_1\ITLRev) \quad &
  R_5 & \quad \ITLTheorem (f\ITLChopstar)\ITLRev \ITLEquiv (f\ITLRev)\ITLChopstar\\
  \end{array}
\]
So we need similar laws for $h(e_1,\ldots,e_n)$ and $(\ITLExists{V} f)$, which we introduce in this paper. These are:
\[ \begin{array}{ll|ll}
R_6 & \quad \ITLTheorem (h(e_1,\ldots,e_n))\ITLRev \ITLEquiv h(e_1\ITLRev,\ldots, e_n\ITLRev) \quad & R_7 & \quad 
\ITLTheorem (\ITLExists{V} f)\ITLRev \ITLEquiv \ITLExists{V} f\ITLRev\\
\end{array}
\]

We also need to show that expressions can be reflected, i.e., we need to investigate $e_i\ITLRev$. Note: in \cite{Moszkowski2014} only propositional variables have been considered and no notion of temporal variables has been defined. In this paper we define the reflection for temporal variables. 

For expressions we have the following reflection laws 
\[
\begin{array}{ll|ll}
\mathit{ER}_0 & \quad \ITLTheorem c\ITLRev = c \quad &
\mathit{ER}_1 & \quad \ITLTheorem V\ITLRev = \ITLFin V\\
\mathit{ER}_2 & \quad \ITLTheorem (\ITLFin V)\ITLRev = V \quad & 
\mathit{ER}_3 & \quad \ITLTheorem (g(e_1,\ldots,e_n))\ITLRev = g(e_1\ITLRev,\ldots, e_n\ITLRev)\\
\end{array}
\]
where $c$ is a constant value of a particular type and $V$ is a variable of a particular type. 

For the temporal variable $\ITLNext V$ we need to introduce a new construct to the syntax of expressions that will serve as a reflected version of $\ITLNext V$.
The temporal variable  $\ITLPrev V$ denotes
the value of $V$ in the pen-ultimate (previous) state.
 The formal semantics is as follows
 \[ 
 \begin{array}{lll}
\eSem{\ITLPrev Q}{\sigma} & = & 
     \left\{\begin{array}{ll}
       \sigma_{\intlen{\sigma}-1}(Q) & \mbox{if } \intlen{\sigma}>0 \\
       \mbox{choose-any-from}(\mathsf{Bool}) & \mbox{ otherwise}
     \end{array}
     \right .\\
\eSem{\ITLPrev A}{\sigma} & = & 
     \left\{\begin{array}{ll}
       \sigma_{\intlen{\sigma}-1}(A) & \mbox{if } \intlen{\sigma}>0 \\
       \mbox{choose-any-from}(\mathbb{Z}) & \mbox{ otherwise}
     \end{array}
     \right .
     \end{array}
\]
The relationship between $(\ITLNext V)$ and $\ITLPrev V$ is captured by the following laws. 
\[\begin{array}{ll|ll}
  \mathit{ER}_4 & \quad\ITLTheorem (\ITLNext V)\ITLRev = \ITLPrev V \quad &
  \mathit{ER}_5 & \quad \ITLTheorem (\ITLPrev V)\ITLRev = \ITLNext V\\
\end{array}
\]
This leads to the following theorem
\begin{theorem}
ITL (extended with $\ITLPrev V$) is closed under reflection.
\end{theorem}

\begin{proof}
The proof of this is done using structural induction on the syntax of ITL. For ITL formulae we use laws $\mathit{R}_0$ -- $\mathit{R}_7$ and for expressions we use laws $\mathit{ER}_0$ -- $\mathit{ER}_5$.
\end{proof}

%

\subsubsection{Temporal reflection laws for derived operators}
We first define the following derived  constructs:
The formula  $\ITLPrev f \ITLdefs f\ITLChop\ITLSkip$ denotes that $f$ holds  previously. The formula $\ITLWeakPrev f \ITLdefs \ITLEmpty \ITLOr \ITLPrev f$  (weak previously) denotes
either an interval of only one state or $f$ holds  previously.
\[\begin{array}{l|l}
  \ITLTheorem \ITLEmpty\ITLRev \ITLEquiv \ITLEmpty \quad & \quad 
  \ITLTheorem \ITLMore\ITLRev \ITLEquiv \ITLMore\\
  \ITLTheorem (\ITLNext f)\ITLRev \ITLEquiv \ITLPrev (f\ITLRev) \quad & \quad 
  \ITLTheorem (\ITLPrev f)\ITLRev \ITLEquiv \ITLNext (f\ITLRev)\\
  \ITLTheorem (\ITLWeakNext f)\ITLRev \ITLEquiv \ITLWeakPrev (f\ITLRev) \quad & \quad
  \ITLTheorem (\ITLWeakPrev f)\ITLRev \ITLEquiv \ITLWeakNext (f\ITLRev)\\
  \ITLTheorem (\ITLSometime f)\ITLRev \ITLEquiv \ITLDi (f\ITLRev) \quad & \quad 
  \ITLTheorem (\ITLDi f)\ITLRev \ITLEquiv \ITLSometime (f\ITLRev)\\
  \ITLTheorem (\ITLAlways f)\ITLRev \ITLEquiv \ITLBi (f\ITLRev) \quad & \quad 
  \ITLTheorem (\ITLBi f)\ITLRev \ITLEquiv \ITLAlways (f\ITLRev)\\
  \ITLTheorem (\ITLDa f)\ITLRev \ITLEquiv \ITLDa (f\ITLRev) \quad & \quad 
  \ITLTheorem (\ITLBa f)\ITLRev \ITLEquiv \ITLBa (f\ITLRev)\\
  \end{array}
  \]
  Observe that the $\ITLSometime$ ($\ITLAlways$) and $\ITLDi$ ($\ITLBi$) are dual wrt to reflection, i.e., the reflection operator relates prefix intervals with suffix intervals and vice versa. 
  Similarly we  have that $\ITLNext$ ($\ITLWeakNext$) and $\ITLPrev$ ($\ITLWeakPrev$) are dual wrt to reflection. 
  

\section{Executability, reflection and reversibility}
\label{sec:execrefrev}
In this section we will discuss the notion of executability. It is used to determine whether an ITL formula represents a programming construct. 
We first formalise the notion of forward executability of a formula which corresponds to generating a sequence of states in a particular fashion: we first generate the first state and then generate the next until the final state is generated. This sequence of state constitutes the behaviour of the system described by the formula. We then investigate the reflection of forward executable formula and this requires the introduction of the notion of backward executability. This notion corresponds to generating a sequence of states but now we first generate the final state  and then generate the previous state until we generate  the first state. This sequence corresponds to the reversed behaviour of the system described by the formula. Forward and backward executability are related by the reflection operator.

\subsection{Forward executability}
\label{subsec:exec}
The intuition of an executable formula (specification) is that it corresponds to a computation, i.e., in our case a sequence of states. Obviously any executable formula needs to be satisfiable. But not every satisfiable formula is executable because we further require it to be ``deterministic''. We will give a formal definition what we mean by this. The executable formula corresponds to programming constructs and 
in Table~\ref{tab:procon} some of these programming constructs are defined. Note: $\ITLFin f$ is different from $\ITLFin V$, the first one is a formula whereas the latter denotes an expression, i.e., is a value. So $\ITLFin$ is overloaded for formulae and expressions. 

\begin{table}[htb]
    \centering
    $
    \begin{array}{lll}
       A = e  &  & \mbox{assignment} \\
       A:=e   & (\ITLNext A) = e &  \mbox{unit assignment} \\
       A\Tassign e & (\ITLFin A) = e & \mbox{temporal assignment} \\
       A \ITLGets e & \ITLBa ( \ITLSkip \ITLImplies A   \Tassign e) & \mbox{gets assignment} \\
       \ITLIf f_0 \ITLThen f_1 \ITLElse f_2 & (f_0 \ITLAnd f_1) \ITLOr (\ITLNot f_0 \ITLAnd
f_2)& \mbox{binary choice} \\
       \ITLLen(0) & \ITLEmpty & \\
       \ITLLen(n+1) & \ITLSkip\ITLChop\ITLLen(n) & \mbox{length of an interval} \\
       \ITLFin f & \ITLAlways (\ITLEmpty \ITLImplies f) & \mbox{in final state $f$ holds} \\
       \ITLInit f & \ITLBi (\ITLEmpty \ITLImplies f) & \mbox{initially $f$ holds} \\
       \ITLHalt f & \ITLAlways (\ITLEmpty \ITLEquiv f)& \mbox{halt when $f$ holds} \\
       \ITLKeep f & \ITLBa (\ITLSkip \ITLImplies f)& \mbox{for all unit intervals $f$ holds} \\
       \ITLWhile f_0 \ITLDo f_1 & (f_0 \ITLAnd f_1)\ITLChopstar \ITLAnd \ITLFin \ITLNot f_0&
        \mbox{while loop} \\
        \ITLExists{V} f & & \mbox{local variable introduction}\\
    \end{array}
    $
    \caption{Tempura constructs}
    \label{tab:procon}
\end{table}

The following definitions are used to determine whether a formula is executable or not. First we define the notion of a value trace of a formula wrt a list of variables. These variables are the ''free variables'' appearing in $f$, i.e., $f$ constrains the values of these variables. Note: 
These definitions and all subsequent theorems have been specified and verified in the Isabelle/HOL \cite{Isabelle} system (library available from \cite{Cauhomepage}).
\begin{definition}
\ 
\begin{itemize}
\item Let $\mathbf{s}$ be a state and let $\overline{\mathsf{v}}$ denote a non-empty list of variables $\mathsf{v}_0,\ldots,\mathsf{v}_n$ and let $\overline{\mathbf{s}(\mathsf{v})}$ denote the corresponding list of values $\mathbf{s}(\mathsf{v}_0), \ldots, \mathbf{s}(\mathsf{v}_n)$ of $\overline{\mathsf{v}}$ in state $\mathbf{s}$. 

\item Let $\mathsf{Spec}$ be a formula and $\sigma$ be an interval and $\Sem{\mathsf{Spec}}{\sigma} = \STrue$ then the value trace of $\mathsf{Spec}$ wrt $\overline{\mathsf{v}}$ is denoted by $\intmap{(\lambda \mathbf{s}. \overline{\mathbf{s}(\mathsf{v})})}{\sigma}$ and defined as
$\overline{\sigma_0(\mathsf{v})}\,\overline{\sigma_1(\mathsf{v})}\ldots \overline{\sigma_{\intlen{\sigma}}(\mathsf{v})}$. 
\end{itemize}
\end{definition}

\begin{example}
  The value trace for $A=0 \ITLAnd A \ITLGets A + 1 \ITLAnd \ITLAlways (B= A*2)$ wrt $(A,B)$ is $(0,0)\ (1,2)\ (2,4)\ (3,6)\ (4,8)\ \ldots$  and it represents how $A$ and $B$ change, $A$ is increased by one and $B$ equals twice $A$ in every state.
\end{example}
The following definition is a constraint on the intervals which satisfy a formula. Only intervals that share a common prefix of the value trace are allowed, 
\begin{definition}\ \\
A formula $\mathsf{Spec}$ has a common prefix value trace wrt a list of variables $\overline{\mathsf{v}}$ denoted by $\ddagger[\mathsf{Spec}]_{\overline{\mathsf{v}}} $  if and only if
for all intervals $\sigma$ and $\sigma'$ if 
$\Sem{\mathsf{Spec}}{\sigma} = \STrue $ and $\Sem{\mathsf{Spec}}{\sigma'} = \STrue$ and $\intlen{\sigma} \leq \intlen{\sigma'}$ then

$(\intmap{(\lambda \mathbf{s}. \overline{\mathbf{s}(\mathsf{v})}) }{ \sigma)} = (\intmap{ (\lambda \mathbf{s}. \overline{\mathbf{s}(\mathsf{v})})}{ (\sigma'_0\ldots\sigma'_{\intlen{\sigma}}))} $.
\end{definition}
In above definition we compare the value trace corresponding to $\sigma$ with the prefix (of length $\intlen{\sigma}$) of the value trace of corresponding to $\sigma'$. The intuition is that the latter is a continuation of the first,i.e.,
the first value trace is a ``beginning'' of the latter value trace.  
The following example illustrates this notion.
\begin{example}\label{ex:commonpfx}
The following are some formula that have a common prefix value trace.
\begin{itemize}
    \item $\ddagger[A=0 \ITLAnd \ITLEmpty]_A$, there is only one possible value trace $0$.
        
   \item $\ddagger[A=0 \ITLAnd A \ITLGets A+1]_A$,
     the possible value traces are
     \[\begin{array}{lll}
       0 \\ 
       0, & 1 \\
       0, & 1, & 2 \\
        \ldots \\
       \end{array}
       \] 
       Each pair of value traces share a common prefix.
       The common prefix value trace of pair $0$ and $0,1$ is $0$ and  of pair $0,1$ and $0,1,2$ is $0,1$. Note that in the latter pair there is another shared prefix $0$ but in the definition it states that we are looking for a prefix that has a length equal to the ``smallest'' of the two. 
       Note we align on the left. 
\end{itemize}    
The following are some formula that have no common prefix value trace. 
\begin{itemize}
    \item $\SNot\ddagger[(A=0 \ITLOr A=1) \ITLAnd \ITLEmpty]_A$, we have two value traces $0$ and $1$, but they do not share a common prefix.
   \item $\SNot\ddagger[A=0 \ITLAnd \ITLSkip ]_A$, we have for instance value traces $0,0$ and $0,1$ but when their length are the same they ought to agree on all values and this does not hold as they disagree in the second state.
   \item $\SNot\ddagger[\ITLSkip]_A $, $A$ does not appear in the formula so values of $A$ are not constrained at all, one has value trace $0,0$ and $1, 0$ and these do not share a common prefix.    
\end{itemize}
\end{example}
The following theorem states that the combination of satisfiability with the notion of common prefix value trace can be used to determine whether a formula is executable or not, i.e., satisfiable and deterministic. 
\begin{theorem}
Let $\mathsf{Spec}$ be a formula and $\overline{\mathsf{v}}$ be a list of variables. 

If $\ITLSat \mathsf{Spec}$   and $\ddagger[\mathsf{Spec}]_{\overline{\mathsf{v}}}$ then
 for all $k \geq 0$
 
$\#\{ (\intmap{(\lambda \mathbf{s}. \overline{\mathbf{s}(\mathsf{v})})}{\sigma})\,  |\, \Sem{\mathsf{Spec}}{\sigma} = \STrue \SAnd \intlen{\sigma} = k  \} \leq 1 $.
\end{theorem}
In above theorem we have that all satisfying intervals of length $k$ will corresponds to at most one value trace. 

The notion of common prefix value trace corresponds to the notion of generating a satisfying interval for a formula but it ''limits'' how  this is achieved, i.e., one proceeds in a forward manner by extending at the right and therefore  no  backtracking will be used. The following definition introduces the notion of forward executability. 
\begin{definition}
Let $\mathsf{Spec}$ be a formula and $\overline{\mathsf{v}}$ a list of variables. 

$\mathsf{Spec}$ is forward executable wrt $\overline{\mathsf{v}}$ denoted by $\dagger[\mathsf{Spec}]_{\overline{\mathsf{v}}}$ if and only if

$\ITLSat \mathsf{Spec}$ and $\ddagger[\mathsf{Spec}]_{\overline{\mathsf{v}}}$.
\end{definition}
In Tempura \cite{tempurabook}, the executable subset of ITL, a formula $\mathsf{Spec}$ is rewritten into a normal form $\ITLInit \mathsf{w}_0 \ITLAnd \ITLWeakNext \mathsf{Spec_0}$.
The $\ITLInit \mathsf{w}_0$ represents the initial state and $\ITLWeakNext \mathsf{Spec_0}$ represents the behaviour of the system from the next state onward but only if there is a next state. This process is repeated for formula $\mathsf{Spec_0}$, i.e., it is rewritten to
$\ITLInit \mathsf{w}_1  \ITLAnd \ITLWeakNext \mathsf{Spec_1}$. 
This process of rewriting into normal form corresponds to our notion of forward executability. This is expressed in the following theorem.
\begin{theorem}\label{th:fexec}
Given formulae $\mathsf{w}$ and $\mathsf{Spec}$
and a list of variables $\overline{\mathsf{v}}$.

If $\dagger[\ITLInit \mathsf{w} \ITLAnd \ITLEmpty]_{\overline{\mathsf{v}}}$ and $\dagger[\mathsf{Spec}]_{\overline{\mathsf{v}}}$ then 
$\dagger[\ITLInit \mathsf{w} \ITLAnd \ITLWeakNext\mathsf{Spec}]_{\overline{\mathsf{v}}}$.
\end{theorem}
In Example~\ref{ex:commonpfx} we have seen that one needs to be careful in adding constructs that limit the length of an interval. The following theorem gives conditions for which it is safe to do so.
\begin{theorem}\label{th:fexecstrengthen} 
Let $\mathsf{Spec}_0$ and $\mathsf{Spec}_1$ be formula and $\overline{\mathsf{v}}$ be a list of variables.

If $\ITLSat \mathsf{Spec}_0 \ITLAnd \mathsf{Spec}_1$ and $\ddagger[\mathsf{Spec}_0]_{\overline{\mathsf{v}}}$ then 
 $\dagger[\mathsf{Spec}_0 \ITLAnd \mathsf{Spec}_1]_{\overline{\mathsf{v}}}$. 
\end{theorem}
In this theorem formula $\mathsf{Spec}_0$ ensures that the values for $\overline{\mathsf{v}}$ are deterministic and formula $\mathsf{Spec}_1$ is used to put extra constraints on the intervals satisfying $\mathsf{Spec}_0$. The $\ITLSat \mathsf{Spec}_0 \ITLAnd \mathsf{Spec}_1$ condition ensures that we have at least one such interval. Examples of such $\mathsf{Spec}_1$ are $\ITLLen(k)$, $\ITLSometime \ITLInit w$ and $\ITLHalt w$. On their own these formulae are not forward executable but combined with a forward executable one they will be.

\subsection{Backward executability}
We now investigate reversing executable specifications. 
We first introduce the laws for the reflection of the programming constructs of Table~\ref{tab:procon}. 
  The programming construct `past assignment' is denoted by $A =: e$ and defined as $(\ITLPrev A) = e$ and `$f$ holds in initial state only' is denoted by $\ITLInitonly f$ and defined as $\ITLBi (\ITLEmpty \ITLEquiv f)$. 
  
Note: in $\ITLPrev f$ and $\ITLPrev V$ the $\ITLPrev$ symbol is overloaded for formulae and expressions. 
  The following are reflection laws for programming constructs of Table~\ref{tab:procon}: 
  \[\begin{array}{ll}
  \ITLTheorem (V = e)\ITLRev \ITLEquiv V \Tassign e\ITLRev &
  \ITLTheorem (V \Tassign e)\ITLRev \ITLEquiv V = e\ITLRev\\
  \ITLTheorem (V := e)\ITLRev \ITLEquiv V =: e\ITLRev &
  \ITLTheorem (V =: e)\ITLRev \ITLEquiv V := e\ITLRev\\
  \ITLTheorem (V \ITLGets e)\ITLRev \ITLEquiv \ITLBa ( \ITLSkip \ITLImplies V = e\ITLRev)\\
  \ITLTheorem (\ITLIf f_0 \ITLThen f_1 \ITLElse f_2)\ITLRev \ITLEquiv \ITLIf f_0\ITLRev \ITLThen f_1\ITLRev \ITLElse f_2\ITLRev\\
  \ITLTheorem (\ITLLen(n))\ITLRev \ITLEquiv \ITLLen(n)\\
  \ITLTheorem (\ITLInit f)\ITLRev \ITLEquiv \ITLFin (f) &
  \ITLTheorem (\ITLFin f)\ITLRev \ITLEquiv \ITLInit (f)\\
  \ITLTheorem (\ITLHalt f)\ITLRev \ITLEquiv \ITLInitonly f\ITLRev &
  \ITLTheorem (\ITLInitonly f)\ITLRev \ITLEquiv \ITLHalt f\ITLRev\\
  \ITLTheorem (\ITLKeep f)\ITLRev \ITLEquiv \ITLKeep (f\ITLRev)\\
  \ITLTheorem ( \ITLInit g \ITLAnd \ITLWhile f_0 \ITLDo f_1)\ITLRev \ITLEquiv \\
  \phantom{\ITLTheorem} \ITLFin(g) \ITLAnd (f_0\ITLRev \ITLAnd f_1\ITLRev)\ITLChopstar \ITLAnd \ITLInit \ITLNot f_0\\
\end{array}
\]
Observe that the $\ITLLen$, $\ITLKeep$ and choice operators are self dual wrt reflection. The $\Tassign$ and assignment operator are dual wrt reflection. So are the unit assignment and past assignment, $\ITLInit$ and $\ITLFin$, $\ITLHalt$ and $\ITLInitonly$.

Additional temporal variables laws allows for replacing temporal variables by other temporal variables when intervals have a specific fixed length
\[\begin{array}{l|l}
  \ITLTheorem \ITLEmpty \ITLImplies (\ITLFin V) = V \\ 
  \ITLTheorem \ITLSkip \ITLImplies (\ITLFin V) = (\ITLNext V) \quad & \quad 
  \ITLTheorem \ITLSkip \ITLImplies (\ITLPrev V) = V\\
\end{array}
\]
Reflection relates the notion of prefix intervals with that of suffix intervals. So we need to introduce the ``mirror image'' of common prefix value traces, i.e. the notion of common suffix value trace. 

The following definition is a constraint on the intervals which satisfy a formula. Only intervals that share a common suffix of the value trace are allowed. 
\begin{definition}\ \\
A formula $\mathsf{Spec}$ has a common suffix value trace wrt a list of variables $\overline{\mathsf{v}}$ denoted by $\natural[\mathsf{Spec}]_{\overline{\mathsf{v}}} $  if and only if
for all intervals $\sigma$ and $\sigma'$ if 
$\Sem{\mathsf{Spec}}{\sigma} = \STrue $ and $\Sem{\mathsf{Spec}}{\sigma'} = \STrue$ and $\intlen{\sigma} \leq \intlen{\sigma'}$ then

$(\intmap{(\lambda \mathbf{s}. \overline{\mathbf{s}(\mathsf{v})}) }{ \sigma)} = (\intmap{ (\lambda \mathbf{s}. \overline{\mathbf{s}(\mathsf{v})})}{ (\sigma'_{\intlen{\sigma'} - \intlen{\sigma}}\ldots\sigma'_{\intlen{\sigma'}}))} $.
\end{definition}
The following example illustrates this notion.
\begin{example}\label{ex:commonsfx}
The following are some formula that have a common suffix value trace.
\begin{itemize}
    \item $\natural[(\ITLFin A)=0 \ITLAnd \ITLEmpty]_A$, there is only one value trace $0$.
        
   \item $\ddagger[\ITLAlways (A = 0) ]_A$,
     the possible value traces are
     \[\begin{array}{lll}
       & & 0 \\
       & 0, & 0 \\
      0, & 0, & 0 \\
      \ldots \\
       \end{array}
       \]
       Each pair of value traces share a common suffix.
\end{itemize}    
The following are some formula that have no common suffix value trace. 
\begin{itemize}
    \item $\SNot\natural[(A=0 \ITLOr A=1) \ITLAnd \ITLEmpty]_A$, we have two value traces $0$ and $1$, but they do not share a common suffix.
   \item $\SNot\natural[(\ITLFin A)=0 \ITLAnd \ITLSkip ]_A$, we have for instance value traces $0,0$ and $1,0$ but when their length are the same they ought to agree on all values and this does not hold as they disagree in the first state.
   \item $\SNot\natural[\ITLSkip]_A $, $A$ does not appear in the formula so values of $A$ are not constrained at all, one has value trace $0,0$ and $0,1$ and these do not share a common suffix.    
\end{itemize}
\end{example}
The following lemma states the relationship between common prefix, common suffix and reflection.
\begin{lemma}
Let $\mathsf{Spec}$ be formula and $\overline{\mathsf{v}}$ be a list of variables. 
\[ \begin{array}{lll}
    \ddagger[\mathsf{Spec}\ITLRev]_{\overline{\mathsf{v}}} &\SIff & \natural[\mathsf{Spec}]_{\overline{\mathsf{v}}}\\
     \natural[\mathsf{Spec}\ITLRev]_{\overline{\mathsf{v}}} &\SIff& \ddagger[\mathsf{Spec}]_{\overline{\mathsf{v}}} \\
     \end{array}
\]
\end{lemma}
For the notion of satisfiability we have the following lemma.
\begin{lemma}
Let $\mathsf{Spec}$ be a formula then 
\[(\ITLSat \mathsf{Spec}\ITLRev) \SIff (\ITLSat \mathsf{Spec}) \] 
\end{lemma}
The following theorem states that the combination of satisfiability with the notion of common suffix value trace can be used to determine whether a formula is deterministic or not, i.e., is backward executable or not.
\begin{theorem}
Let $\mathsf{Spec}$ be a formula and $\overline{\mathsf{v}}$ be a list of variables. 

If $\ITLSat \mathsf{Spec}$   and $\natural[\mathsf{Spec}]_{\overline{\mathsf{v}}}$ then
for all $k \geq 0$

$\#\{ (\intmap{(\lambda \mathbf{s}. \overline{\mathbf{s}(\mathsf{v})})}{\sigma})\, |\, \Sem{\mathsf{Spec}}{\sigma} = \STrue \SAnd \intlen{\sigma} = k  \} \leq 1 $.
\end{theorem}
The notion of common suffix value trace 
corresponds to notion of generating a satisfying interval for a formula but it ''limits'' how  this is achieved, i.e., one proceeds in a backward manner. The following definition introduces the notion of backward executability. 
\begin{definition}
Let $\mathsf{Spec}$ be a formula and $\overline{\mathsf{v}}$ a list of variables.

$\mathsf{Spec}$ is backward executable wrt to $\overline{\mathsf{v}}$ denoted by $\flat[\mathsf{Spec}]_{\overline{\mathsf{v}}}$ if and only if

$\ITLSat \mathsf{Spec}$ and $\natural[\mathsf{Spec}]_{\overline{\mathsf{v}}}$.
\end{definition}
In Tempura we have unfortunately no rules for backward execution. But we can define a mirror image of Theorem~\ref{th:fexec}, i.e., the normal form would be $\ITLFin \mathsf{w} \ITLAnd \ITLWeakPrev \mathsf{Spec}$. So we first generate the last state of the interval and then proceed to determine the previous state if there is any. 
\begin{theorem}
Given the formulae $\mathsf{w}$ and $\mathsf{Spec}$ and list of variables $\overline{\mathsf{v}}$.

If $\flat[\ITLFin \mathsf{w} \ITLAnd \ITLEmpty]_{\overline{\mathsf{v}}}$ and $\flat[\mathsf{Spec}]_{\overline{\mathsf{v}}}$ then  $\flat[\ITLFin \mathsf{w} \ITLAnd \ITLWeakPrev \mathsf{Spec}]_{\overline{\mathsf{v}}}$. 
\end{theorem}
The following theorem is similar to Theorem~\ref{th:fexecstrengthen}. 
\begin{theorem}
Let $\mathsf{Spec}_0$ and $\mathsf{Spec}_1$ be formula and $\overline{\mathsf{v}}$ be a list of variables.

If $\ITLSat \mathsf{Spec}_0 \ITLAnd \mathsf{Spec}_1$ and $\natural[\mathsf{Spec}_0]_{\overline{\mathsf{v}}}$ then $\flat[\mathsf{Spec}_0 \ITLAnd \mathsf{Spec}_1]_{\overline{\mathsf{v}}}$. 
\end{theorem}

\subsection{Reversing the effects of bad computations} 
In the introduction we have seen that we are interested in formulae of the form   $\mathsf{Good}\ITLChop\mathsf{Bad}\ITLChop(\mathsf{Bad})\ITLRev$. We now investigate under  which conditions can we forward execute $\mathsf{Bad}\ITLChop\mathsf{Bad}\ITLRev$.
The chop operator is non-deterministic if  the length of  $\mathsf{Bad}$ is left unspecified, i.e., generally we have 
$\SNot\dagger[\mathsf{Bad}\ITLChop\mathsf{Bad}\ITLRev]_{\overline{v}}$. 
However, we can use Theorem~\ref{th:fexecstrengthen} to strengthen $\mathsf{Bad}$ to $\mathsf{Bad} \ITLAnd \ITLLen (k)$. We similarly strengthen the $\mathsf{Bad}\ITLRev$ to $\mathsf{Bad}\ITLRev \ITLAnd \ITLLen(k)$ in order to ensure that we undone that specific bad computation  $\mathsf{Bad} \ITLAnd \ITLLen(k)$. Note that $(\mathsf{Bad} \ITLAnd \ITLLen(k))\ITLRev $ is equivalent to $\mathsf{Bad}\ITLRev \ITLAnd \ITLLen(k)$, this follows from the reflection laws. 

The following theorem gives the conditions necessary to ``undo'' a bad computation.
\begin{theorem}
Let $\mathsf{Spec}$  be a formula and $\overline{\mathsf{v}}$ be a list of variables. 

If $\ITLSat \mathsf{Spec} \ITLAnd \ITLLen k$ and $\ddagger[\mathsf{Spec}]_{\overline{\mathsf{v}}}$ and $\natural[\mathsf{Spec}]_{\overline{\mathsf{v}}}$ then 
\begin{enumerate}
    \item $\dagger[(\mathsf{Spec} \ITLAnd \ITLLen (k))\ITLChop(\mathsf{Spec}\ITLRev \ITLAnd \ITLLen (k))]_{\overline{\mathsf{v}}}$ and
    \item $\flat[(\mathsf{Spec}\ITLRev \ITLAnd \ITLLen (k))\ITLChop(\mathsf{Spec} \ITLAnd \ITLLen (k))]_{\overline{\mathsf{v}}}$. 
\end{enumerate}

\end{theorem}
Notice that $\mathsf{Spec}$ needs to have both a common prefix value trace and a common suffix value trace. In the first case we proceed in a forward manner while in the second case in a backward manner.

\section{Results and Discussion}
\label{sec:results}

We have shown that first order ITL is closed under reflection. This is an extension of the work of \cite{Moszkowski2014} where reflection of propositional ITL is discussed. We then investigated the reversibility of executable ITL specifications. Executable specifications allow for the testing and simulation of specifications in that the satisfying behaviour of those specifications is generated.  We first formalised the notion of forward executability which corresponds to the way the Tempura tool \cite{tempurabook,Cauhomepage} generates these behaviours. The tool rewrites the specification into a normal form $\ITLInit w \ITLAnd \ITLWeakNext f$ where $\ITLInit w$ specifies the initial state of the behaviour and $\ITLWeakNext f$ corresponds to the behaviour from the next state onward if there is such a next state. We have shown that the rewrite mechanism of Tempura indeed preserves forward executability of specifications.

We then investigated the reversibility of executable specifications. We introduced the notion of backward executability of specifications, this corresponds to mechanism of first generating the last state of the behaviour and then generate the previous behaviour. The required normal form would be $\ITLFin w \ITLAnd \ITLWeakPrev f$. We have shown that in order to reverse a bad computation/behaviour of an executable specification it needs to be both forward and backward executable.  

The notions of common prefix value trace $\ddagger[\mathsf{Spec}]_{\overline{v}}$ and common suffix value trace $\natural[\mathsf{Spec}]_{\overline{v}}$ are related to Allen's logic operators ``starts'' and ``finishes'' \cite{alleninterval} that relate two sequences of entities, i.e.,  in our case sequences of values.

The mechanism proposed in this paper only requires the current state to achieve reversibility so there is no need of the storing of the history of the original executable specification as used in the literature \cite{5391327,KP2014,KUHN2018}. The intuitive approach of saving everything and using this to restore to a desired position is named checkpointing. One variation is to store the initial state and to simply restore to this. No intermediate states can be restored without additional forward re-execution. Another variation of this, named full checkpointing, records the entire state at each intermediate step 
and so allows immediate restoration to any previous position \cite{KP2014}. The concern here is of the amount of information saved, much of which will not be changed by a single step. The periodic or incremental checkpoint variation reduces the memory usage, but relies on restoring to a previous checkpoint and re-executing forward to reach the desired position \cite{KP2014}. In contrast, our method allows reverting to any previous state and uses a low amount of memory for this as only the current state is needed and this is always available, no extra storage is needed.

The type of reversibility we have modelled is backtracking. Since we do not have a notion of causality, we cannot model causally consistent or out-of-causal-order reversibility \cite{KUHN2018}. We also do not control when reversibility is used, modelling a form of uncontrolled reversibility \cite{LaneseMS12}. This relies on a separate mechanism for controlling forward and backwards execution. If the control is integrated with the mechanism for reversibility, we would either have controlled \cite{LaneseMS12} or locally controlled reversibility \cite{KUHN2018}.

\section{Conclusion and Future Work}
\label{sec:conclusion}

First order ITL is a flexible notation for specifying properties and behaviours of systems. Most imperative programming constructs are denoted by formulae in ITL. We have used the reflection operator for the specification of reversed behaviour of systems. It is shown that ITL is closed under this reflection operator which means that we can specify its reverse for any ITL formula. We have presented an extensive list of reflection laws that help in the construction of the reverse of an ITL formula. We have shown that when an ITL formula is forward and backward executable then one can indeed reverse its behaviour. 

Future work consists of adding the backward execution mechanism to the Tempura tool.  The reflection and reversal of event-based programs is another area of interest. In an event-based program, a trigger event causes a chain of reactions by a system. The occurrence of a trigger can not be reversed but the reaction by the system can be reversed. However, this reaction might include other triggers that will set of other chains of reactions. Determining this chain of reactions and reversing its effects are some of the challenges that need to be addressed.





\bibliography{reversible}
\bibliographystyle{splncs04}
\end{document}